\RequirePackage{fix-cm}
\documentclass[smallextended]{svjour3}       
\smartqed  
\usepackage{graphicx}
\usepackage{amssymb}
\usepackage{longtable}
\usepackage{amsmath}
\usepackage{amsfonts}
\begin{document}

\title{A Note on Cyclic Codes from  APN Functions
}


\author{Chunming Tang  \and Yanfeng Qi
\and Maozhi Xu 
}


\institute{Chunming Tang \at
              School of Mathematics and Information, China West Normal University, Sichuan Nanchong, 637002, China\\
           \and
           Yanfeng Qi \at
              LMAM, School of Mathematical Sciences,
              Peking University, Beijing, 100871,
              and Aisino Corporation Inc., Beijing, 100195, China\\
           \and
            Maozhi Xu \at
           LMAM, School of Mathematical Sciences, Peking University, Beijing, 100871, China\\
}

\date{Received: date / Accepted: date}

\maketitle

\begin{abstract}
Cyclic codes, as linear block error-correcting codes in coding theory,  play a vital role and have wide applications. Ding in
\cite{D} constructed a number of classes of
cyclic codes from almost perfect nonlinear (APN)
functions and planar functions over finite fields and
presented ten open problems on cyclic codes from highly
nonlinear functions. In this paper, we consider two open problems involving the inverse APN functions $f(x)=x^{q^m-2}$ and
the Dobbertin APN function
$f(x)=x^{2^{4i}+2^{3i}+2^{2i}+2^{i}-1}$. From the calculation of
linear spans and the minimal polynomials of  two
sequences generated by  these two classes of APN functions, the dimensions  of the corresponding cyclic codes are determined and lower bounds  on the minimum weight of these cyclic codes are presented.
Actually, we present a framework for the minimal polynomial and linear span of the sequence $s^{\infty}$ defined by
$s_t=Tr((1+\alpha^t)^e)$, where $\alpha$ is a primitive element in $GF(q)$.
These techniques can also be applied into other open problems in
\cite{D}.
\end{abstract}

\section{Introduction}
As a subclass of linear codes, cyclic codes with efficient
encoding and decoding algorithms \cite{CHI,F,P}, have been widely applied into data storage systems and communication systems.   Hence, cyclic codes attract much attention and have been widely
studied. Many research results have been made \cite{BM,C,HP,LW}.

Throughout this paper, let $p$ be a prime and $q$ be a power of $p$. Let
$r=q^m$, where $m$ is a positive integer.
A linear $[n,k,d]$ code $\mathcal{C}$ over finite field  $GF(q)$
is  a $k$ dimension subspace of $GF(q)^n$ with
minimum distance $d$. The linear code $\mathcal{C}$ is called a cyclic code if $(c_0,c_1,\cdots,c_{n-1})\in \mathcal{C}$ implies that $(c_{n-1},c_0,\cdots,c_{n-2})\in \mathcal{C}$. Let
$gcd(n,q)=1$. Then there exists the following correspondece:
\begin{eqnarray*}
\pi: GF(q)^n&\longrightarrow & GF(q)[x]/(x^n-1)\\
(c_0,c_1,\ldots, c_{n-1})&\longmapsto&
c_0+c_1x+\ldots+c_{n-1}x^{n-1}.
\end{eqnarray*}
Then a codeword $(c_0,c_1,\cdots,c_{n-1})\in \mathcal{C}$
can be identified with  the polynomial
$\label{(1)}c_0+c_1x+\cdots+c_{n-1}x^{n-1}\in GF(p)[x]/(x^n-1)$. Further, $\mathcal{C}$ is a cyclic code if and only if $
\pi(\mathcal{C})$ is an ideal of $GF(p)[x]/(x^n-1)$. Since
$GF(p)[x]/(x^n-1)$ is a principal ideal ring, there exists a
unique monic polynomial $g(x)|(x^n-1)$ satisfying $
\pi(\mathcal{C})=\langle g(x) \rangle$. We also write
$\mathcal{C}=\langle g(x) \rangle$. Then
$g(x)$ is called the generator polynomial of
$\mathcal{C}$  and $h(x)=(x^n-1)/g(x)$ is called the parity-check
polynomial of $\mathcal{C}$.

Sequences over finite field $GF(q)$ can be utilized to construct
cyclic codes \cite{D}, that is, the minimal polynomial of a periodic sequence can be used as the generator polynomial of a cyclic code. Let $s^\infty=(s_t)_{t=0}^\infty$ be a sequence
over $GF(q)$ with period $q^m-1$. Then there exists the smallest positive integer $l$, satisfying
\begin{align}\label{(2)}-c_0s_i=c_1s_{i-1}+c_2s_{i-2}+
\cdots+c_ls_{i-l},~(i\geq l),\end{align}
where $c_0=1,c_1,\cdots,c_l\in GF(q)$ and  $c_i(0\leq i \leq l)$ are uniquely determined.
The unique polynomial
\begin{align}\label{(3)}\mathbb{M}_s(x)=c_0+c_1x+\cdots+c_lx^l\in GF(q)[x]\end{align}
is called the minimal polynomial of the sequence $s^\infty$ and  $l$ is the linear span (also called linear complexity) of  $s^\infty$, which is denoted as $\mathbb{L}_s$. Generally, any polynomial as a constant multiple of $\mathbb{M}_s(x)$ is also
a minimal polynomial of   $s^\infty$. Further, $\mathbb{M}_s(x)$ and $\mathbb{L}_s$ are given by \cite{LN}
\begin{align}
\mathbb{M}_s(x)&=\frac{x^{q^m-1}-1}{gcd(x^{q^m-1}
-1,S^{q^m-1}(x))},\label{(4)}\\
\mathbb{L}_s&=q^m-1-deg(gcd(x^{q^m-1}-1,S^{q^m-1}(x))),
\label{(5)}\end{align}
where
$$S^{q^m-1}(x)=s_0+s_1x+\cdots+s_{q^m-2}x^{q^m-2}\in GF(q)[x].$$
Hence, the cyclic code with the generator polynomial $\mathbb{M}_s(x)$ is
\begin{align}\label{(6)}\mathcal{C}_s=\langle \mathbb{M}_s(x)\rangle.
\end{align}
Then, $\mathcal{C}_s$ is an $[q^m-1,q^m-1-\mathbb{L}_s,d]$ cyclic code. Consequently, the determination of $\mathbb{M}_s(x)$ and
$d$ is essential for studying these cyclic code, which is the main problem in
\cite{D} and this paper.

Generally, it is difficult to determine $\mathbb{M}_s(x)$ for
a sequence $s^\infty$. DingÔÚ\cite{D} considered a special
class of sequences with the form
\begin{align}
\label{(7)}s_t=Tr((1+\alpha^t)^e),
\end{align}
where $\alpha$ is a primitive element of $GF(r)$, $Tr(\cdot)$ is the trace function from $GF(r)$ to
$GF(q)$ and $e$ is an exponent making the
monomial $f(x)=x^e$ an almost perfect nonlinear function or a
planar function. In some special cases, Ding determined $\mathbb{M}_s(x)$, which induced a number of classes of cyclic codes.  Further, ten open problems for the
determination of $\mathbb{M}_s(x)$ were proposed.

In this paper, some results on $\mathbb{M}_s(x)$ are given first.
Then we solve the following two open problems presented by Ding.
Two classes of APN functions used to induce cyclic codes are

{\rm (1)} The inverse APN function over finite field $GF(r)$\cite{HR}: $f(x)=x^{q^m-2}$;

{\rm  (2)} The Dobbertin APN function over finite field $GF(2^{5i})$ \cite{DO1}: $f(x)=x^{2^{4i}+2^{3i}+2^{2i}+2^{i}-1}$.

The minimal polynomial of  the inverse APN function with $p=2$ was
solved in \cite{D,SD}. We generalize the result for a general odd  prime $p$, which gives the corresponding cyclic code. In fact, this technique can be useful for other
open problems presented by Ding.

\section{Preliminary}
In this section, some notations and results are introduced.
Let $Z_N=\{0,1$, $\cdots$, $N-1\}$ with the integer addition modulo $N$
and integer multiplication modulo $N$ operation for a positive integer $N$.
Let the $q-$cyclotomic coset containing $j$ modulo $N$ with the action
$x\longmapsto qx\mod N$ be defined by
$$
C_j=\{j,qj, q^2j,\ldots, q^{l_j-1}\} \subset Z_N,
$$
where $l_j$ is the smallest positive integer such that
$q^{l_j}j\equiv j \mod N$ and $l_j$ is called the size of
$C_j$.
For a sequence over finite field $GF(q)$ with period $q^m-1$,
we have the following result on the linear span and minimal polynomial \cite{AB}.
\begin{lemma}\label{GF}
Let $s^\infty$ be a sequence with period $q^m-1$ over finite field
$GF(q)$ , which has a unique expansion of the form
$$
s_t=\sum_{i=0}^{q^m-2}c_i\alpha^{it}, \forall t\geq 0,
$$
where $\alpha$ is a generator of $GF(q^m)^*$, and $c_i\in GF(q^m)$. Let $I=\{i:c_i\neq 0\}$, then the minimal polynomial of $s^\infty$ is
$$
\mathbb{M}_s(x)=\prod_{i\in I}(1-\alpha^ix),
$$
and the linear span $\mathbb{L}_s$ of $s^\infty$ is $|I|$.
\end{lemma}
From this lemma, in order to determine the minimal polynomial and linear span, we just need to compute the expansion of $s_t$.
The Lucas's theorem \cite{L} can be a useful tool in the computation.
\begin{theorem}\label{NM}
Let $N=n_tq^t+n_{t-1}q^{t-1}+\cdots+n_1q+n_0$ and $M=m_tq^t+
m_{t-1}q^{t-1}+\cdots+m_1q+m_0$, then
$$
\binom{N}{M}\equiv \binom{n_{t}}{m_{t}}\binom{n_{t-1}}{m_{t-1}}\cdots
\binom{n_{0}}{m_{0}} \mod p.
$$
\end{theorem}
Two most important classes of  nonlinear functions in cryptography are almost perfect
nonlinear (APN) functions and perfect nonlinear (or planar) functions \cite{BD,NY}, which can be used to construct cyclic codes \cite{CCD,CCZ,CDY,D}.
A function  $f:GF(r)\longrightarrow GF(r)$ is called a planar
function if
$$
\max_{a\in GF(r)^*,b\in GF(r)}\#(\{x\in GF(r):f(x+a)-f(x)=b\})=1.
$$
And a function $f:GF(r)\longrightarrow GF(r)$
is called an almost perfect nonlinear (APN) function
if
$$\max_{a\in GF(r)^*,b\in GF(r)}\#(\{x\in GF(r):f(x+a)-f(x)=b\})=2.$$

The APN functions considered in this paper are

$\bullet$ The inverse APN function \cite{HR}: $f(x)=x^{q^m-2}$;\\

$\bullet$ The Dobbertin APN function \cite{DO1}: $f(x)=x^{2^{4i}+2^{3i}+2^{2i}+2^{i}-1},m=5i.$

\section{The minimal polynomial of sequence from binomial functions}
In this section, we analyze the determination of the minimal polynomial
of the sequence $s^{\infty}$ with $s_t=Tr((1+\alpha^t)^e)$, where
$0\leq e \leq q^m-2$. Further, Let $n=q^m-1$ for the rest of the paper.
\begin{align}
s_t&=Tr((1+\alpha^t)^e)\notag\\
             &=Tr( \sum_{i=0}^{q^m-2}\binom{e}{i}\alpha^{it})\notag\\
            &=\sum_{i=0}^{q^m-2}\binom{e}{i}\sum_{j=0}^{m-1}
            \alpha^{q^{j}it}\notag\\
            &=\sum_{i=0}^{q^m-2}(\sum_{j=0}^{m-1}\binom{e}{q^ji\mod n}\mod p)\alpha^{it}\notag\\
            &=\sum_{i=0}^{q^m-2}C_{e,q,m}(i)\alpha^{it}\label{(8)}
\end{align}
where $C_{e,q,m}(i)=\sum_{j=0}^{m-1}\binom{e}{q^ji\mod n}\mod p.$

Define
\begin{align}
Supp(C_{e,q,m})=\{i:0\leq i<n,C_{e,q,m}(i)\neq 0\}.\label{(9)}
\end{align}
From the definition of $C_{e,q,m}(i)$, we have
$$
C_{e,q,m}(q^ji)=C_{e,q,m}(i).
$$
Then there exists a subset $\widetilde{Supp}(C_{e,q,m})$ of
$Supp(C_{e,q,m})$ satisfying
\begin{align}
Supp(C_{e,q,m})=\bigcup_{i\in \widetilde{Supp}(C_{e,q,m})}\{q^ji \mod n:j=0,1,\cdots,m-1\}.\label{(10)}
\end{align}
Further, we can choose $\widetilde{Supp}(C_{e,q,m})$ with the condition that if  $i$,$i'$ $\in \widetilde{Supp}(C_{e,q,m})$ and $j=1,\cdots,m-1$,
$n\nmid (i'-q^ji)$. Then
\begin{align}
s_t&=Tr((1+\alpha^t)^e)\notag\\
             &=\sum_{i\in Supp(C_{e,q,m})}C_{e,q,m}(i)\alpha^{it}\notag\\
            &=\sum_{i\in \widetilde{Supp}(C_{e,q,m})}\sum_{i'\in C_{i}}C_{e,q,m}(i')\alpha^{i't}\label{(11)}
\end{align}
where $C_i$ is the $q$-cyclotomic coset containing $i$ modulo $n=q^m-1$,  that is , $C_i=\{q^ji\mod n:j=0,1,\cdots,m-1\}$.

From the above discussion, we have the following theorem.
\begin{theorem}\label{LM}
Let $s^\infty$ be a sequence over $GF(q)$ defined by $s_t=Tr((1+\alpha^t)^e)$.
Then the linear span $\mathbb{L}_s$ of $s^\infty$ is $\#(Supp(C_{e,q,m}))$, and the minimal polynomial of $s^\infty$ is
$$
\mathbb{M}_s(x)=\prod_{\gamma\in Supp(C_{e,q,m})}(x-\gamma^{-1})=\prod_{\gamma\in \widetilde{Supp}(C_{e,q,m})}m_{\gamma^{-1}}(x),
$$
where $m_{\gamma^{-1}}(x)$ is the minimal polynomial of $\gamma^{-1}$ over $GF(q)$.
\end{theorem}
\begin{proof}
From Identity (\ref{(11)}) and Lemma \ref{GF}, this theorem follows.
\end{proof}
\section{Cyclic codes from APN function: $x^{q^m-2}$}

Note that $\binom{N-1}{M}=\frac{N-M}{N}
\binom{N}{M}$, then
\begin{align}
C_{e,q,m}(i)&\equiv \sum_{i'\equiv q^ji\mod n,j\in Z_m}\binom{e}{i'}\notag\\
&\equiv \sum_{i'\equiv q^ji\mod n,j\in Z_m}\frac{e+1-i'}{e+1}\binom{e+1}{i'}\notag\\
&\equiv \sum_{i'\equiv q^ji\mod n,j\in Z_m}(1+i')\binom{e+1}{i'}\mod p\label{(12)}
\end{align}
Let $i=\sum_{j=0}^{m-1}i_jq^j$, then $qi \mod n=i_{m-1}+i_0q+\cdots+i_{m-2}q^{m-1}$. From
$e+1=\sum_{j=0}^{m-1}(q-1)q^j$ and Lucas's theorem, we have
\begin{align*}
\binom{e+1}{qi\mod n}&\equiv \binom{q-1}{i_{m-2}}\cdots\binom{q-1}{i_{0}}\binom{q-1}{i_{m-1}}\\
&\equiv \binom{q-1}{i_{m-1}}\cdots\binom{q-1}{i_{1}}\binom{q-1}{i_{0}}\\
&\equiv \binom{e+1}{i} \mod p
\end{align*}
For $i'=q^ji\mod n (j \in Z_m)$,  $\binom{e+1}{i'}\equiv \binom{e+1}{i} \mod p$. Then we get
\begin{align}
C_{e,q,m}(i)&\equiv \binom{e+1}{i}\sum_{i'\equiv q^ji\mod n,j\in Z_m}(1+i')\notag\\
&\equiv \binom{e+1}{i}(\sum_{j=0}^{m-1}i_j+m) \mod p\label{(13)}
\end{align}
where $i_j\equiv q^ji \mod n$.
Let $s_q(i)=\sum_{j=0}^{m-1}i_j\label{(14)}$, then
\begin{align}
Supp(C_{e,q,m})=\{i\in Z_n:s_q(i)+m\not\equiv0 \mod p\}.\label{(15)}
\end{align}
Further, we have
\begin{align}
\#(Supp(C_{e,q,m}))=q^m(1-\frac{1}{p}).
\label{(16)}
\end{align}
\begin{theorem}\label{LM1}
Let $s^\infty$ be a sequence defined by $s_t=Tr((1+\alpha^t)^{q^m-2})$.
Then the linear span $\mathbb{L}_s$ of $s^\infty$ is $q^m(1-\frac{1}{p})$, and the minimal polynomial of
$s^\infty$ is
\begin{align}
\mathbb{M}_s(x)=\prod_{\mbox{\tiny$\begin{array}{c}i\in Z_n\\s_q(i)+m\not\equiv0\mod p\end{array}$}}(x-\alpha^{-i}),\label{M1}
\end{align}
where $s_q(i)=\sum_{j=0}^{m-1}i_j$.
\end{theorem}
\begin{proof}
From Identities (\ref{(15)}), (\ref{(16)}) and Theorem \ref{LM}, this theorem follows.
\end{proof}
\begin{theorem}\label{CK1}
The cyclic code $\mathcal{C}_s$, induced by the sequence in
Theorem $\ref{LM1}$ is an
$[n,\frac{q^m}{p}-1,d]$ code, whose generator polynomial is given
by Identity $(\ref{M1})$. Further, $d\geq\max \{2p-1,\frac{q(p-1)}{p}+1\}$.
\end{theorem}
\begin{proof}
The dimension of $\mathcal{C}_s$ can be obtained by Theorem \ref{LM1}. Then we consider the lower bound of minimum weight $d$.  Note that
the weight distribution of the cyclic code generated by $\mathbb{M}_s(x)$ is
the same as that of the cyclic code generated by the reciprocal polynomial of $\mathbb{M}_s(x)$. The reciprocal polynomial of $\mathbb{M}_s(x)$ has zero points
$\alpha^{j}$,  where
$j\in \{q-p+1,q-p+2,\cdots,q+p-2\}$. From the BCH bound, we have $d\geq 2p-1$. Further,  the reciprocal polynomial of $\mathbb{M}_s(x)$ has zero points $\alpha^{pj+1}$, where $j\in \{0,1,\cdots,\frac{q(p-1)}{p}-1\}$. From the  Hartmann-Tzeng bound \cite{HT},
we have $d\geq \frac{q(p-1)}{p}+1$. Hence, $d\geq\max \{2p-1,\frac{q(p-1)}{p}+1\}$.
\end{proof}
\begin{example}
Let $q=3$ and $m=2$. Let $\alpha$ be a primitive element of $GF(3^2)$ satisfying
$\alpha^2+2\alpha+2=0$. The the corresponding cyclic code $\mathcal{C}_s$ is an $[8,2,6]$ code, whose generator polynomial is
\begin{align*}
\mathbb{M}_s(x)&=x^{6}+2x^{5}+2x^{4}+2x^{2}+x+1.
\end{align*}
And its dual  is an $[8,6,2]$ cyclic code.
\end{example}
\begin{example}
Let $q=3$ and $m=3$. Let $\alpha$ be a primitive element of  satisfying  $\alpha^3+2\alpha+1=0.$ The corresponding cyclic code
$\mathcal{C}_s$ is an $[26,8,10]$ code, whose generator polynomial is
\begin{align*}
\mathbb{M}_s(x)=&x^{18}+2x^{16}+2x^{15}+x^{14}+x^{12}+x^{11}+x^{10}+\\
&x^9+x^8+x^7+x^6+x^4+2x^3+2x^2+1.
\end{align*}
And its dual is an
$[26,18,4]$ cyclic code.
\end{example}
\section{Cyclic codes from the APN function: $x^{2^{4i}+2^{3i}+2^{2i}+2^i-1}$}

Let $m=5i$ and $e=2^{4i}+2^{3i}+2^{2i}+2^i-1$. The binary sequence of  $e$ is
\begin{align}
\underline{e}=\underbrace{00\cdots1}_i
\underbrace{00\cdots1}_i\underbrace{00\cdots1}_i
\underbrace{00\cdots0}_i\underbrace{11\cdots1}_i\label{(17)},
\end{align}
which has $m$ bits.

For an integer $x\in\{0,1,\cdots,n-1\}$, the binary expansion of $x$ is $\sum_{j=0}^{m-1}x_j2^j$, and the corresponding binary sequence  is  $\underline{x}=
x_{m-1}x_{m-1}\cdots x_0$ with $m$ bits. Then $\underline{x}_i$ and $x_i$ are the same. Let $St$ be the set of all the binary sequences with $m$ length except
$1^m=\underbrace{11\cdots1}_m$. Let $G$ be the  group with elements of cyclically shift transformations on  $St$. Then $G$ is a
cyclic code with order $m$. Let $\tau\in G$ and $\tau(b_{m-1}b_{m-2}\cdots b_0)=b_0b_{m-1}\cdots b_1$, then $G=\langle \tau \rangle$. Note that there is a bijection map between $Z_n$ and $St$. Then $\tau$ is corresponding to an action $x\longmapsto q^{m-1}x$ over $Z_n$, that is, a $q-$cyclotomic action on $Z_n$ is corresponding to a cyclic action on $St$.

\begin{definition}\label{D1}
Let $\underline{x},\underline{y}\in St$. It is said that $\underline{x}$ is covered by $\underline{y}$ or $\underline{y}$ is a cover of  $\underline{x}$, if for $\underline{x}_i=1$, then $\underline{y}_i=1$.  It is denoted by that $\underline{x}\preceq \underline{y}$.
\end{definition}
\begin{definition}\label{D2}
Let $\underline{x}\in St$. It is said that $\underline{x}$ is cyclically covered by  $\underline{e}$, if there exists an element $\sigma \in G$ satisfying $\sigma(\underline{x})\preceq \underline{e}$. If such $\sigma$ does not exist, then $\underline{x}$ is not cyclically covered by
$\underline{e}$.
\end{definition}
\begin{definition}\label{D3}
Let $\underline{x}\in St$. the sequence
$\underline{x}$ is called an even sequence if $\#\{\sigma \in G| \sigma(\underline{x})\preceq \underline{e}\}$ is even. Otherwise,
$\underline{x}$ is an odd sequence.
\end{definition}

\begin{lemma}\label{S}
Let $m=5i$ and $e=2^{4i}+2^{3i}+2^{2i}+2^i-1$. Let $p=q=2$, $r=2^m$ and $x\in Z_n$. Then $x\in Supp(C_{e,q,m})$ if and only if $\underline{x}$ is an odd sequence.
\end{lemma}
\begin{proof}
From Lucas's theorem, we have
\begin{align*}
C_{e,q,m}(x)&=\sum_{x'=q^jx,j\in Z_m}\binom{e}{x'}\\
&=\sum_{\underline{x}'=\sigma(\underline{x}),\sigma\in G}\prod_{k=0}^{m-1} \binom{\underline{e}_k}{\underline{x}'_k}\\
&=\#(\{\sigma\in G:\sigma(\underline{x})\preceq \underline{e}\})\mod 2.
\end{align*}
Hence, $C_{e,q,m}\neq 0$ if and only if $\underline{x}$ is an odd sequence.
\end{proof}

\begin{lemma}\label{Rl}
Let $l\geq 2$ and $\underline{y}=\underbrace{00\cdots0}_i
\underbrace{00\cdots0}_i\underbrace{00\cdots0}_i
\underbrace{00\cdots0}_i\underbrace{1*\cdots*1}_l
\underbrace{0\cdots0}_{i-l}$, where
$*$ is any element in $\{0,1\}$. Then
$\#\{\sigma(\underline{y})|\sigma \in G\}=m$.
\end{lemma}
\begin{proof}
We just need to prove that if $\sigma\neq id$, then $\sigma(\underline{y})\neq \underline{y}$. Let $\sigma=\tau^j$, where $0<j<m$.

{\rm (1)} If $0<j \leq i-l$, then
$$(\sigma(\underline{y}))_{i-l-j}=\underline{y}_{i-l}=1.$$
Further, $0\leq i-l-j<i-l$. Then
$$\underline{y}_{i-l-j}=0.$$
Hence,
$$(\sigma(\underline{y}))_{i-l-j}\neq \underline{y}_{i-l-j}.$$

{\rm (2)} If $i-l<j \leq m-l$, then
$$(\sigma(\underline{y}))_{m+i-l-j}=\underline{y}_{i-l}=1.$$
Further, $i\leq m+i-l-j\leq m-1$. Then
$$\underline{y}_{m+i-l-j}=0.$$
Hence,
$$(\sigma(\underline{y}))_{m+i-l-j}\neq \underline{y}_{m+i-l-j}.$$

{\rm (3)} If $m-l<j < m$, then
$$(\sigma(\underline{y}))_{m+i-1-j}=\underline{y}_{i-1}=1.$$
Further, $i-1< m+i-1-j< i+l-1$. Then
$$\underline{y}_{m+i-1-j}=0.$$
Hence,
$$(\sigma(\underline{y}))_{m+i-1-j}\neq \underline{y}_{m+i-1-j}.$$

Consequently, when $0< j < m$, $\sigma(\underline{y})=\tau^j(\underline{y})\neq \underline{y}$.
This lemma follows.
\end{proof}
If $l\geq 2$, denote the set of all the elements like
 $\underline{y}$ in Lemma $\ref{Rl}$ by $\mathcal{R}^1_l$, that is,
\begin{align}\label{R111}
\mathcal{R}^1_l=\{\underline{y}\in St:\underline{y}=\underbrace{00\cdots0}_{4i}
\underbrace{1*\cdots*1}_{l}\underbrace{00\cdots0}_{i-l}\}.
\end{align}
\begin{lemma}\label{Rl1}
Let $\underline{y} \in \mathcal{R}^1_{l}$. Let $\sigma=\tau^j$, where $0\leq j \leq i-1$. Then $\sigma(\underline{y})\preceq \underline{e}$ if and only if $0\leq j \leq i-l$. Further,
$\underline{y}$ is an odd sequence if and only if $i-l+1$ is odd.
\end{lemma}
\begin{proof}
From $0\leq j\leq i-l$, $\sigma(\underline{y})\preceq \underline{e}.$ Further, we have
$$(\sigma(\underline{y}))_{i-1-j}=\underline{y}_{i-1}=1,$$
$$(\sigma(\underline{y}))_{i-l-j}=\underline{y}_{i-l}=1.$$
From $\sigma(\underline{y})\preceq \underline{e}$,
$$\underline{e}_{i-1-j}=\underline{e}_{i-l-j}=1.$$
Hence, the $(i-1-j)$-th element of $\underline{y}$ is $1$. Further, from the $(i-1-j)$-th element, the
$(l-1)$-th element is also $1$. Then
$$i-1-j (~mod~~ m~) \in \{i-1,i-2,\cdots,l-1\},$$
Hence,
$$j \in \{0,1,\cdots,i-l\}.$$
This lemma follows.
\end{proof}
\begin{lemma}\label{R1}
The set of all the odd sequences in
$\bigcup_{2\leq l \leq i}\mathcal{R}^1_l$ is
$$\mathcal{R}^1=\bigcup_{2\leq l \leq i,l\equiv i\mod 2}\mathcal{R}^1_l.$$
Further,
$$
\#(\mathcal{R}^1)=
\begin{cases}\frac{2^i-1}{3},&i\text{~is even}, \cr \frac{2^i-2}{3},&i\text{~is odd}.\end{cases}
$$
\end{lemma}
\begin{proof}
From Lemma \ref{Rl1}, the set of all the odd sequence in $\bigcup_{2\leq l \leq i}\mathcal{R}^1_l$ is
$$
\mathcal{R}^1=\bigcup_{2\leq l \leq i,l\equiv i\mod 2}\mathcal{R}^1_l.
$$
Further, $\#(\mathcal{R}^1_l)=2^{l-2}$. When $l\neq l'$, $\mathcal{R}^1_l\cap \mathcal{R}^1_{l'}=\emptyset$.

If $i$ is even,
$$
\#(\mathcal{R}^1)=2^0+2^2+2^4+\cdots +2^{i-2}=\frac{2^i-1}{3}.
$$
If $i$ is odd,
$$
\#(\mathcal{R}^1)=2^1+2^3+2^5+\cdots +2^{i-2}=\frac{2^i-2}{3}.
$$
This lemma follows.
\end{proof}

\begin{lemma}\label{R2}
Let $\mathcal{R}^2$ be the set of all the elements $\underline{y}$ in $St$ satisfying the following conditions:

{\rm (1)} If $j\geq i$ and $j\notin \{2i,3i,4i\}$, $\underline{y}_j=0$;

{\rm (2)} There exists $j\in \{2i,3i,4i\}$ satisfying $\underline{y}_j=1$;

{\rm (3)} $\#(\{j:0\leq j \leq i-1,\underline{y}_j=1\})\geq 2$.

For any $\underline{y}\in \mathcal{R}^2$, let $\sigma\in G$ satisfying $\sigma(\underline{y})\preceq \underline{e}$. Then $\sigma=1$.
\end{lemma}
\begin{proof}
From Condition (3), choose $j_1$ and $j_2$ satisfying
$$
0\leq j_1< j_2\leq i-1,\underline{y}_{j_1}=\underline{y}_{j_2}=1.
$$
Let $\sigma=\tau^j$ and $\sigma(\underline{y})\preceq \underline{e}$, then
$$\tau^j(\underline{y})_{j_1-j}=\underline{y}_{j_1}=1=
\underline{e}_{j_1-j},$$
$$\tau^j(\underline{y})_{j_2-j}=\underline{y}_{j_2}=1=
\underline{e}_{j_2-j}.$$
Further,
$$1\leq (j_2-j)-(j_1-j)\leq i-1,$$
From the sequence $\underline{e}$,
$$j_2-j (~mod~~ m~)\in \{1,2,\cdots,i-2,i-1\},$$
that is,
$$j=j_2-(i-1),j_2-(i-2),\cdots,j_2-1.$$
Further,
$$-(i-2)\leq j\leq i-2.$$
From Condition (2), just choose $j_0\in\{2i,3i,4i\}$ sastisfying $\underline{y}_{j_0}=1$. Note that
$$\tau^j(\underline{y})_{j_0-j}=\underline{y}_{j_0}=1.$$
Hence,
$$\underline{e}_{j_0-j}=1.$$
We can have $j=0$, that is, $\sigma=1$.
\end{proof}

\begin{lemma}\label{R27}
Let $\mathcal{R}^2$ be defined above. Then $\#(\mathcal{R}^2)=7(2^i-i-1)$.
\end{lemma}
\begin{proof}
This lemma can be obtained by the definition of $\mathcal{R}^2$.
\end{proof}
\begin{lemma}\label{R4}
Let $\mathcal{R}_4=\{\underline{y}\in St:\underline{y}\preceq \underline{e},wt(\underline{y})=4,\underline{y}\notin \mathcal{R}^1\bigcup\mathcal{R}^2\}$.
Then $\#(\mathcal{R}_4)=i$. Further, let $\sigma\in G$ satisfying $\sigma(\underline{y})\preceq \underline{e}$ for any $\underline{y}\in \mathcal{R}_4$, then $\sigma=1.$
\end{lemma}
\begin{proof}
From the definition of $\mathcal{R}_4$, $\underline{y}\in \mathcal{R}_4$ if and only if
$$
\{j:i\leq j\leq 5i-1,\underline{y}_j=1\}=\{2i,3i,4i\},\#(\{j:0\leq j \leq i-1,\underline{y}_j=1\})=1.
$$
Hence,
$$
\#(\mathcal{R}_4)=i.
$$
Let $\underline{y}\in \mathcal{R}_4$ and $\sigma(\underline{y})\preceq \underline{e}$. Then
$$\sigma(\underline{y})\in \mathcal{R}_4.$$
Let $\sigma=\tau^j$. Then
$$
\#(\{(2i-j)\mod 5i,(3i-j) \mod 5i,(4i-j) \mod 5i\}\bigcap \{2i,3i,4i\})\geq 2.
$$
Hence,
$$
j\in\{0,i,-i,2i,-2i\}.
$$
When $j=\pm i,\pm 2i$, $\tau^j(\underline{y})\npreceq \underline{e}$. Hence, $j=0$,that is, $\sigma=1.$
\end{proof}

\begin{lemma}\label{R3}
Let $\mathcal{R}_3=\{\underline{y}\in St:\underline{y}\preceq \underline{e},wt(\underline{y})=3,\underline{y}
\text{~is odd},\underline{y}\notin \mathcal{R}^1\bigcup\mathcal{R}^2\}$. Then
$\#(\mathcal{R}_3)=3(i-1)$. Further, let $\sigma\in G$ satisfying $\sigma(\underline{y})\preceq \underline{e}$  for any $\underline{y}\in \mathcal{R}_3$. Then $\sigma=1$.
\end{lemma}
\begin{proof}
From the definition of $\mathcal{R}_3$, an element $\underline{y}$ in $\mathcal{R}_3$ satisfies one of the following
three conditions:

{\rm (1)} $\underline{y}_{2i}=\underline{y}_{3i}=\underline{y}_{4i}=1$;

{\rm (2)} $\#(\{j:\underline{y}_j=1,j\in \{2i,3i,4i\}\})=2,\underline{y}_0=1$;

{\rm (3)}$\#(\{j:\underline{y}_j=1,j\in \{2i,3i,4i\}\})=2,\#(\{j:\underline{y}_j=1,1\leq j \leq i-1\})=1$.

We will prove that $\underline{y}$ satisfying one of Condition (1) and (2) is even, and $\underline{y}$ in Condition (3) is odd.

$\underline{y}$  in Condition (1) or (2) can fall into the following four cases:

{\rm (a)}  $\underline{y}_{2i}=\underline{y}_{3i}=\underline{y}_{4i}=1$. It is denoted by $\underline{y}^{(1)};$

{\rm (b)} $\underline{y}_{3i}=\underline{y}_{4i}=\underline{y}_{0}=1$.
It is denoted by $\underline{y}^{(2)};$

{\rm  (c)}$\underline{y}_{2i}=\underline{y}_{3i}=\underline{y}_{0}=1$.
It is denoted by $\underline{y}^{(3)};$

{\rm (d)} $\underline{y}_{4i}=\underline{y}_{2i}=\underline{y}_{0}=1$. It is denoted by $\underline{y}^{(4)}.$

$\underline{y}^{(1)}$ satisfies that $\{\sigma\in G:\sigma(\underline{y}^{(1)})\preceq \underline{e}\}=\{\tau^0,\tau^{4i}\}$. Hence, $\underline{y}^{(1)}$ is even. Since $\tau^{4i}(\underline{y}^{(1)})=\underline{y}^{(2)}$,
$\underline{y}^{(2)}$ is also even.

$\underline{y}^{(3)}$ satisfies $\{\sigma\in G:\sigma(\underline{y}^{(3)})\preceq \underline{e}\}=\{\tau^0,\tau^{3i}\}$. Hence, $\underline{y}^{(3)}$ is even. Since $\tau^{3i}(\underline{y}^{(3)})=\underline{y}^{(4)}$,
$\underline{y}^{(4)}$ is also even.

Consequently, $\underline{y}$  in Condition (1) or (2) is an even sequence.

$\underline{y}$  in Condition (3) satisfies
$$
\{\sigma\in G:\sigma(\underline{y})\preceq \underline{e}\}=\{\tau^0\}.
$$
Hence, $\underline{y}$ in Condition (3) is odd. Further, all the sequences in $\mathcal{R}_3$ fall into sequences in Condition (3). Then, $\#(\mathcal{R}_3)=3(i-1).$\\
The proof of $\sigma=1$ is similar to that in Lemma \ref{R4}.
\end{proof}

\begin{lemma}\label{R22}
Let $\mathcal{R}_2=\{\underline{y}^{(1)},\underline{y}^{(2)}\}$,
where
$$\underline{y}^{(1)}=\underbrace{00\cdots 01}_i\underbrace{00\cdots01}_i
\underbrace{00\cdots00}_i
\underbrace{00\cdots00}_i
\underbrace{00\cdots00}_i,$$
$$
\underline{y}^{(2)}=\underbrace{00\cdots 01}_i\underbrace{00\cdots00}_i\underbrace{00\cdots01}_i
\underbrace{00\cdots00}_i\underbrace{00\cdots00}_i.
$$
Let $\mathcal{R}_2'=\{\underline{y}\in St:\underline{y}\preceq \underline{e},wt(\underline{y})=2,\underline{y}\text{~is odd},
\underline{y}\notin \mathcal{R}^1\bigcup\mathcal{R}^2\}$. Then $\underline{y}^{(1)}$ is not cyclically equivalent to
$\underline{y}^{(2)}$. Further, $\mathcal{R}_2'=$
$
\{\underline{y}^{(1)}$,$\tau^i(\underline{y}^{(1)})$,$\tau^{4i}
(\underline{y}^{(1)})$,$\underline{y}^{(2)}$,$\tau^{2i}
(\underline{y}^{(2)})$,$\tau^{4i}(\underline{y}^{(2)})\}$.
\end{lemma}
\begin{proof}
It is obvious that $\underline{y}^{(1)}$ and $\underline{y}^{(2)}$ are not cyclically equivalent.
Note that
$$\{\sigma\in G:\sigma(\underline{y}^{(1)})\preceq \underline{e}\}=\{\tau^0,\tau^i,\tau^{4i}\},$$
$$\{\sigma\in G:\sigma(\underline{y}^{(2)})\preceq \underline{e}\}=\{\tau^0,\tau^{2i},\tau^{4i}\}.$$
Then $\underline{y}^{(1)},\underline{y}^{(2)}\in \mathcal{R}_2'$.
From the definition of $\mathcal{R}_2'$,
$\tau^i(\underline{y}^{(1)})4$, $\tau^{4i}(\underline{y}^{(1)})$,
$\tau^{2i}(\underline{y}^{(2)})$, $\tau^{4i}(\underline{y}^{(2)})\in \mathcal{R}_2'$.
To complete the proof, we just show that $\underline{y}$ in the following three conditions does not lie in $\mathcal{R}_2'$.

{\rm (1)} $wt(\underline{y})=2,\underline{y}_{4j}=\underline{y}_j=1,0<j\leq {i-1}$. It is denoted by $\underline{y}^{(3)};$

{\rm (2)} $wt(\underline{y})=2,\underline{y}_{2j}=\underline{y}_j=1,0<j\leq {i-1}$. It is denoted by $\underline{y}^{(4)};$

{\rm (3)} $wt(\underline{y})=2,\underline{y}_{3j}=\underline{y}_j=1,0<j\leq {i-1}$. It is denoted by $\underline{y}^{(5)}.$

Note that
$$\{\sigma\in G:\sigma(\underline{y}^{(3)})\preceq \underline{e}\}=\{\tau^{0},\tau^{3i+j}\},$$
$$\{\sigma\in G:\sigma(\underline{y}^{(4)})\preceq \underline{e}\}=\{\tau^{0},\tau^{i+j}\},$$
$$\{\sigma\in G:\sigma(\underline{y}^{(5)})\preceq \underline{e}\}=\{\tau^{0},\tau^{2i+j}\},$$
Hence, $\underline{y}^{(3)}$, $\underline{y}^{(4)}$ and $\underline{y}^{(5)}$ are even.
Consequently, this lemma follows.
\end{proof}

If $i$ is an odd integer, any sequence with weight $1$ is an even sequence. If $i$ is even, any sequence with weight $1$ is an odd sequence. However, the zero sequence  is odd if $i$ is odd and it is even if $i$ is even. For convenience, some notations are given below.
\begin{align}
\mathcal{R}_1=\begin{cases}
\{\underbrace{00\cdots0}_{m-1}1\},&i\text{~is even}, \cr \emptyset,&i\text{~is odd}.\end{cases}\label{(18)}
\end{align}
\begin{align}
\mathcal{R}_0=\begin{cases}\emptyset,&i\text{~is even}, \cr \{\underbrace{00\cdots0}_m\},&i\text{~is odd}.\end{cases}
\label{(19)}
\end{align}
\begin{align}
\mathcal{R}=\mathcal{R}^1\bigcup
\mathcal{R}^2\bigcup\mathcal{R}_4\bigcup\mathcal{R}_3
\bigcup\mathcal{R}_2\bigcup\mathcal{R}_1\bigcup\mathcal{R}_0.
\label{(20)}
\end{align}
Then we have the following lemma.
\begin{lemma}\label{R}
Let $\mathcal{R}$ be defined above. Then

{\rm (1)} $\#(\mathcal{R})=\begin{cases} \frac{22}{3}(2^i-1)-3i,&i\text{~is even}\cr \frac{22}{3}(2^i-2)-3i+7 ,&i\text{~is odd}\end{cases};$

{\rm (2)} Let   $\sigma\in G$ satisfy $\sigma(\underline{x})\in \mathcal{R}$ for some $\underline{x}\in \mathcal{R}$. Then $\sigma=\tau^0;$

{\rm (3)} Let $\underline{x}\in \mathcal{R}$ and $x\neq 0^m$. Then $\#(\{\sigma(\underline{x}):\sigma\in G\})=m;$

{\rm (4)} If $x\in Supp(C_{e,q,m})$, there exists $\sigma\in G$ satisfying $\sigma(\underline{x})\in \mathcal{R}.$

Further,
$$
\widetilde{Supp}(C_{e,q,m})=\{x\in Z_n:\underline{x}\in \mathcal{R}\},
$$
and
$$
\#(Supp(C_{e,q,m}))=\begin{cases} m[\frac{22}{3}(2^i-1)-3i],&i\text{~is even}\cr m[\frac{22}{3}(2^i-2)-3i+6]+1 ,&i\text{~is odd}\end{cases}.
$$
\end{lemma}
\begin{proof}
Result (1) can be obtained by Lemma \ref{R1}, Lemma \ref{R27}, Lemma \ref{R4},  Lemma \ref{R3} and definitions of
$\mathcal{R}_2$, $\mathcal{R}_1$, $\mathcal{R}_0$.

Result (2) can be obtained by definitions of $\mathcal{R}^1$, $\mathcal{R}^2$, $\mathcal{R}_4$,
$\mathcal{R}_3$, $\mathcal{R}_2$, $\mathcal{R}_1$ and $\mathcal{R}_0$.

Result (3) can be obtained by Lemma \ref{R1}, Lemma \ref{R2}, Lemma \ref{R4},  Lemma \ref{R3} and definitions of $\mathcal{R}_2$, $\mathcal{R}_1$, $\mathcal{R}_0$.

Result (4) can be obtained by the definition of $\mathcal{R}$.

Finally,  the result on $\widetilde{Supp}(C_{e,q,m})$ can be obtained by Result (2) and (4). The result on $\#(Supp(C_{e,q,m}))$ can be obtained by Result (1) and (3).
\end{proof}
\begin{theorem}\label{LM2}
Let $s^\infty$ be a sequence defined by $s_t=Tr((1+\alpha^t)^{2^{4i}+2^{3i}+2^{2i}+2^i-1})$. Then the linear span $\mathbb{L}_s$ of $s^\infty$ is
\begin{align}
\mathbb{L}_s=\begin{cases}m[\frac{22}{3}(2^i-1)-3i],
&i\text{~is even}, \cr m[\frac{22}{3}(2^i-2)-3i+6]+1,&i\text{~is odd}.\end{cases}
\label{L2}
\end{align}
And the minimal polynomial of $s^\infty$ is
\begin{align}
\mathbb{M}_s(x)=\begin{cases}\prod_{i\in \mathcal{R}}m_{\alpha^{-i}}(x),&i\text{~is even},\cr (x-1)\prod_{i\in \mathcal{R},i\neq 0^m}m_{\alpha^{-i}}(x),&i\text{~is odd}.\end{cases}\label{M2}\end{align}
where $m_{\alpha^{-i}}(x)$ is the minimal polynomial of $\alpha^{-i}$.
\end{theorem}
\begin{proof}
This theorem follows from Lemma \ref{R} and Theorem \ref{LM}.
\end{proof}
\begin{theorem}\label{CK2}
The binary cyclic code $\mathcal{C}_s$ induced by the sequence
$s^{\infty}$ in Theorem ${\ref{LM2}}$ is an $[2^m-1,2^m-1-\mathbb{L}_s,d]$ code, whose generator polynomial is $\mathbb{M}_{s}(x)$ in Identity $(\ref{M2})$.
$\mathbb{L}_s$ is given by Identity $(\ref{L2})$. Further,
$$
d\geq \begin{cases}2^i+1,~i\equiv 0 \mod 2\cr 2^i+2,~i\equiv 1 \mod 2\end{cases}.
$$
\end{theorem}
\begin{proof}
The dimension of $\mathcal{C}_s$ can be obtained by Theorem $\ref{LM2}$ and the definition of $\mathcal{C}_s$. Then we will give the lower bound of minimum weight.

Note that
the weight distribution of the cyclic code generated by $\mathbb{M}_s(x)$ is
the same as that of the cyclic code generated by the reciprocal polynomial of $\mathbb{M}_s(x)$. From Lemma $\ref{R}$, the reciprocal polynomial of $\mathbb{M}_s(x)$ has zero points $\alpha^{2^{4i}+2^{3i}+2^{2i}+j}$, where
$j\in \{0,1,2,...2^{i}-1\}$.
From the BCH bound, we have $d\geq 2^i+1$. If $m$ is odd, then $\mathcal{C}_s$ is an even-weight code and $d\geq 2^i+2$.
\end{proof}
\begin{example}
Let $q=2$ and $m=5$. Let $\alpha$ be a primitive element of $GF(2^5)$ satisfying
$\alpha^5+\alpha^2+1=0$. Then the corresponding cyclic code $\mathcal{C}_s$ is an $[31,15,8]$ code, whose generator polynomial is
\begin{align*}
\mathbb{M}_s(x)&=x^{16}+x^{14}+x^{13}+x^{10}+x^{9}+x^{8}+
x^{7}+x^6+x^5+x^2+x+1.
\end{align*}
And its dual is an $[31,16,7]$ cyclic code.
\end{example}
\begin{example}
Let $q=2$ and $m=10$. Let $\alpha$ be a primitive element of $GF(2^{10})$ satisfying  $\alpha^{10}+\alpha^6+\alpha^5+\alpha^3+\alpha^2+\alpha+1=0$.
Then the corresponding cyclic code $\mathcal{C}_s$ is an $[1023,863]$ code. And its dual is an $[1023,160]$ cyclic code.
\end{example}

\section{Conclusion}

In this paper, we solve two open problems presented by Ding, which are on cyclic codes induced by the inverse APN functions and the Dobbertin APN functions. Further, we determine the dimensions and generator polynomials of these cyclic codes, which is from the determination of linear spans and minimal polynomials of sequences. The techniques we used in this paper are general and can be utilized in other open problems presented by Ding in \cite{D}. Actually, we present a framework for determining the minimal polynomial of the sequence $s^{\infty}$ defined by $s_t=Tr((1+\alpha^t)^e)$.

\begin{acknowledgements}
This work was supported by the National Natural Science Foundation of China (Grant Nos. 61272499, 10990011). Yanfeng Qi acknowledges support
from Aisino Corporation Inc.
\end{acknowledgements}

\end{document}